%% file: emptyrect1.tex
\documentclass[12pt]{article}

\usepackage{latexsym,color,graphicx,amsmath,amssymb,amsthm,fullpage}

\def\E{{\cal E}}

\def\M{{\cal M}}
\def\bd{{\partial}}

\def\mx{{\rm max}}

\newtheorem{theorem}{Theorem}[section]
\newtheorem{lemma}[theorem]{Lemma}

\begin{document}

\begin{titlepage}

\title{Finding the Maximal Empty Rectangle Containing a Query
Point\thanks{%
Work by Haim Kaplan was partially supported by Grant 2006/204 from
the U.S.--Israel Binational Science Foundation, and by Grant 822/10
from the Israel Science Fund.
Work by Micha Sharir was partially
supported by NSF Grant CCR-08-30272,
by Grant 2006/194 from the U.S.--Israel Binational Science Foundation,
by Grant 338/09 from the Israel Science Fund,
and by the Hermann Minkowski--MINERVA Center for Geometry at Tel Aviv
University.
}}

\author{
Haim Kaplan\thanks{%
School of Computer Science, Tel Aviv University,
Tel Aviv 69978, Israel. E-mail: {\tt haimk@post.tau.ac.il }}
\and
Micha Sharir\thanks{%
School of Computer Science, Tel Aviv University,
Tel Aviv 69978, Israel, and
Courant Institute of Mathematical Sciences,
New York University, New York, NY 10012, USA.
E-mail: {\tt michas@post.tau.ac.il }}%
}

\maketitle

\vspace*{-0.5cm}
\begin{abstract}
Let $P$ be a set of $n$ points in an axis-parallel rectangle $B$ in
the plane. We present an $O(n\alpha(n)\log^4 n)$-time algorithm to
preprocess $P$ into a data structure of size $O(n\alpha(n)\log^3
n)$, such that, given a query point $q$, we can find, in $O(\log^4
n)$ time, the largest-area axis-parallel rectangle that is contained
in $B$, contains $q$, and its interior contains no point of $P$.
This is a significant improvement over the previous solution of
Augustine {\em et al.} \cite{qmex}, which uses slightly
superquadratic preprocessing and storage.
\end{abstract}

\end{titlepage}

\section{Introduction}
\label{sec:intro}

Let $P$ be a set of $n$ points in a fixed axis-parallel rectangle
$B$ in the plane. A {\em $P$-empty rectangle} (or just an empty
rectangle for short) is any axis-parallel rectangle that is
contained in $B$ and its interior does not contain any point of $P$.
We consider the problem of preprocessing $P$
into a data structure so that, given a query point $q$, we can
efficiently find the largest-area $P$-empty rectangle containing $q$.
This problem arises in electronic design automation, in the
context of the design and verification of physical layouts of
integrated circuits (see, e.g., \cite[Chapter 9]{Ullman:1984}).

The largest-area $P$-empty rectangle containing $q$ is a
{\em maximal empty rectangle}, namely, it is a $P$-empty
rectangle not contained in any other $P$-empty rectangle.
Each side of a maximal empty rectangle abuts a point of $P$
or an edge of $B$. See Figure~\ref{mx0} for an illustration.
Maximal empty rectangles arise in the enumeration of ``maximal
white rectangles'' in image segmentation \cite{baird90}.

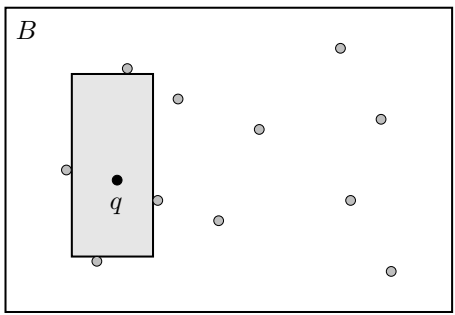
\begin{figure}[htb]
\begin{center}
\input{mx0.pstex_t}
\caption{A maximal $P$-empty rectangle containing $q$.}
\label{mx0}
\end{center}
\end{figure}

The problem considered here can be formulated in more general
settings by considering other classes of shapes that have to contain
the query point (and be $P$-empty), and other kinds of bounding
regions. To the best of our knowledge this problem was first
introduced by Augustine \emph{et al.}~\cite{qmex}, who studied the
case where the regions containing the query point are disks and the
case where these regions are axis-parallel rectangles. For the case
of disks they give a data structure that requires $O(n^2)$ space,
$O(n^2\log n)$
 preprocessing time, and can answer a query in $O(\log^2 n)$
time. For the case of rectangles (the one also considered here) they
give a data structure whose storage and preprocessing time are both
$O(n^2\log n)$, and the query time is $O(\log n)$.

\medskip

\paragraph{Our result.}
We significantly improve the result of Augustine \emph{et al.}~for
the case of axis-parallel rectangles, in terms of storage and
preprocessing costs. Specifically, we present a data structure that
requires $O(n\alpha(n)\log^3 n)$ space and can be used to find the
largest-area $P$-empty rectangle containing a query point $q$ in
$O(\log^4 n)$ time. The structure can be constructed in
$O(n\alpha(n) \log^4 n)$ time. Here $\alpha(n)$ is the slowly
increasing inverse Ackermann function.

In a nutshell, our algorithm computes all the maximal $P$-empty
rectangles and preprocesses them into a data structure which is then
searched with the query point. A major problem that one faces is
that the number of maximal $P$-empty rectangles can be quadratic in
$n$ (see, e.g., Figure \ref{stair}), so we cannot afford to compute
them explicitly (this issue was ignored in \cite{qmex}).

One of the main ingredients of our solution developed to overcome
this difficulty and significant in itself, is a technique for
handling \emph{partial inverse Monge matrices} (see
\cite{SMAWK87,Klawe92,KlaweK90}). Specifically, we observe that the
areas of certain subsets of maximal empty rectangles can be arranged
in a matrix satisfying the \emph{inverse Monge property}; see below
for details. The structure of standard Monge (and inverse Monge)
matrices supports linear or near-linear algorithms (in the number of
rows and columns of the matrix) for finding all the row maxima (or
minima) in such a matrix. These algorithms usually avoid an explicit
representation of the matrix, and instead compute the row maxima or
minima by accessing only a small number of entries, each of which
can be retrieved in $O(1)$ time. For full matrices only a linear
number of entries is needed~\cite{SMAWK87}, and for certain kinds of
structured partial matrices an almost linear number of entries
suffices~\cite{Klawe92,KlaweK90}.

We extend these basic techniques, and develop a data structure that
supports efficient maxima \emph{queries} in certain submatrices of a
partial inverse Monge matrix. In a typical query of this kind we
specify a row and a range of columns, and seek the maximum element
of this row within this range. In another kind of queries, we
specify a prefix of the rows and a prefix of the columns, and seek
the maximum in the submatrix formed by these prefixes. A variation
of our data structure has already been applied in a recent maximum
flow algorithm for planar graphs \cite{BKMNW11}. Our data structure
can be extended for general submatrix queries and  is likely to find
additional applications.

\medskip

\paragraph{Related work.}
An easier problem that has been studied more extensively is that
of finding the largest-area $P$-empty axis-parallel rectangle
contained in $B$.
Notice that the largest $P$-empty \emph{square} is easier to compute,
because its center is a Voronoi vertex in the $L_\infty$-Voronoi
diagram of $P$ (and of the edges of $B$), which can be found in
$O(n\log n)$ time \cite{ChewD85,LS87}. There have been several
studies on finding the largest-area bounded maximal empty
rectangle~\cite{Atallah1986,Chazelle:1986,Naamad84}; the fastest
algorithm to date, by Aggarwal and Suri~\cite{Aggarwal:1987},
takes $O(n\log^2 n)$ time and $O(n)$ space.
Nandy \emph{et al.}~\cite{Nandy94} show how to find the largest-area
axis-parallel empty rectangle avoiding a set of polygonal
obstacles, within the same time bounds.
Boland and Urrutia~\cite{Boland01findingthe} present an algorithm
for finding the largest-area axis-parallel rectangle inside an
$n$-sided simple polygon in $O(n\log n)$ time.
Chaudhuri \emph{et al.}~\cite{Chaudhuri:2003} give an algorithm
to find the largest-area $P$-empty rectangle, with no restriction on
its orientation, in $O(n^3)$ time.

The variant studied in this paper, of finding the largest $P$-empty
rectangle containing a query point is newer and, as already
mentioned, the only previous study of this problem that we are aware
of is by Augustine {\em et al.}\ \cite{qmex}.

\section{Preliminaries}

We assume that the points of $P$ are in general position, so that
(i) no two points have the same $x$-coordinate or the same
$y$-coordinate, and (ii) all the maximal $P$-empty rectangles
have distinct areas.

One of the auxiliary structures that we use is a two-dimensional
segment tree, which stores certain subsets of $P$-maximal empty
rectangles. Here is a brief review of the structure, provided for
the sake of completeness. Let $\M$ be a set of $N$ axis-parallel
rectangles in the plane. We first construct a standard segment tree
$S$~\cite{Dutchbook}  on the $x$-projections of the rectangles in
$\M$. This is a balanced binary search tree whose leaves correspond
to the intervals between the endpoints of the $x$-projections of the
rectangles. The \emph{span} of a node $v$ is the minimal interval
containing all intervals corresponding to the leaves of its subtree.
We store a rectangle $R$ at each node $v$ such that the
$x$-projection of $R$ contains the span of $v$ but does not contain
the span of the parent of $v$. The tree has $O(N)$ nodes, each
rectangle is stored at $O(\log N)$ nodes, and the size of the
structure is thus $O(N\log N)$. All the rectangles containing a
query point $q$ must be stored at the nodes on the search path of
the $x$-coordinate of $q$ in the tree.

For each node $u$ of $S$ we take the set $\M_u$ of rectangles stored
at $u$, and construct a secondary segment tree $S_u$, storing the
$y$-projections of the rectangles of $\M_u$. The total size and the
preprocessing time of the resulting two-dimensional segment tree
is $O(N\log^2 N)$. We can retrieve all rectangles containing a query
point $q$ by traversing the search path $\pi$ of
(the $x$-coordinate of) $q$ in the primary tree, and then by
traversing the search paths of (the $y$-coordinate of) $q$ in each
of the secondary trees associated with the nodes along $\pi$.
The rectangles stored at the secondary nodes along these paths are
exactly those that contain $q$. If we store at each secondary node
only the rectangle of largest area among those assigned to that node,
we can easily find the largest-area rectangle of $\M$ containing a
query point, in time $O(\log^2N)$.
Storing only one rectangle at each secondary node reduces the size
of the segment tree to $O(N\log N)$, but the preprocessing time
remains $O(N\log ^2 N)$.

This simple-minded solution will be efficient only when the size of
$\M$ is linear or nearly linear in $n$. Unfortunately, as already
noted, in general the number of maximal empty rectangles can be
quadratic in the input size, so for most of them we will need an
additional, implicit representation. The structure that we will use
for this purpose is a \emph{partial inverse Monge matrix}, so we
first provide a brief background on Monge matrices.

\paragraph{Monge matrices: A brief review.}
A matrix $M$ is a \emph{Monge matrix} (resp., an \emph{inverse Monge
matrix}) if for every pair of rows $i < j$ and every pair of columns
$k < \ell$ we have $M_{i k} + M_{j \ell} \le M_{i \ell} + M_{j k}$
(resp., $M_{i k} + M_{j \ell} \ge M_{i \ell} + M_{j k}$). A matrix
is {\em totally monotone} if, for every pair of rows $i < j$ and
every pair of columns $k < \ell$, $M_{ik} \le M_{i \ell}$ implies
$M_{jk} \le M_{j \ell}$. It is easy to verify that an inverse Monge
matrix is totally monotone. Aggarwal \emph{et al.}~\cite{SMAWK87}
gave an algorithm for finding all row maxima in a totally monotone
$m\times n$ matrix in $O(m+n)$ time. A \emph{partial totally
monotone matrix} is a matrix some of whose entries are undefined,
but it satisfies the total monotonicity condition for every pair of
rows $i < j$ and every pair of columns $k < \ell$ for which all four
entries $M_{i k}$, $M_{i \ell}$, $M_{j k}$, and $M_{j \ell}$, are
defined. Klawe and Kleitman \cite{KlaweK90} give an $O(n\alpha(m) +
m)$-time algorithm for finding all row maxima of {\em staircase
totally monotone matrices}. These are partial totally monotone
matrices in which the defined part of each row is contiguous,
starting from the first column, and the defined part of each row is
not smaller than the defined part of the preceding row. Klawe
\cite{Klawe92} also present an $O(n\log\log m + m)$-time algorithm
for finding all row maxima in \emph{skyline totally monotone
matrices}, where the defined part of each column is contiguous
starting from the bottommost row. All these algorithms for partially
totally monotone matrices can find row minima instead of row maxima
within the same time bound.

\section{The Data Structure}
\label{sec:rect}

Let $P$ be a set of $n$ points inside an axis-parallel rectangular
region $B$ in the plane. Recall that our goal is
to preprocess $P$ into a data
structure, so that, given a query point $q\in B$, we can
efficiently find the largest-area axis-parallel $P$-empty
rectangle containing $q$ and contained in $B$.

\subsection{Maximal rectangles with edges on the boundary of $B$}
\label{sec:boundary}

Let $e_t$, $e_b$, $e_\ell$, and $e_r$ be the top, bottom, left, and
right edges of $B$, respectively. We classify the maximal $P$-empty
rectangles within $B$ according to the number of their edges that
touch the edges of $B$. We show that there are only $O(n)$ maximal
$P$-empty rectangles with at least one edge on $\bd B$. We
precompute these rectangles and store them in a two-dimensional
segment tree $S$, as described above. At query time we find the
rectangle of largest area among those special ``anchored''
rectangles that contain the query point $q$, by searching with $q$
in $S$. (The segment tree $S$ will also store additional rectangles
that will arise in later steps of the construction; see below for
details.)

Here is the classification and analysis of maximal $P$-empty
rectangles $R$ with at least one edge on $\bd B$.

\paragraph{(i) Three edges of $R$ lie on $\bd B$.}
It is easy to verify that there are only four such rectangles,
one for each triple of edges of $B$.

\paragraph{(ii) Two adjacent edges of $R$ lie on $\bd B$.}
Suppose, without loss of generality, that the top and right edges of
$R$ lie on $e_t$ and $e_r$, respectively. The other two edges of $R$
must be supported by a pair of \emph{maxima} of $P$ (that is, points
$p\in P$ for which no other point $q\in P$ satisfies $x_q > x_p$ and
$y_q > y_p$), consecutive in the sorted order of the maxima (by
their $x$- or $y$-coordinates). Since there are $O(n)$ pairs of
consecutive maxima, the number of anchored rectangles of this kind
is also $O(n)$. See Figure~\ref{mx2a}. The other three situations
are handled in a fully symmetric manner.

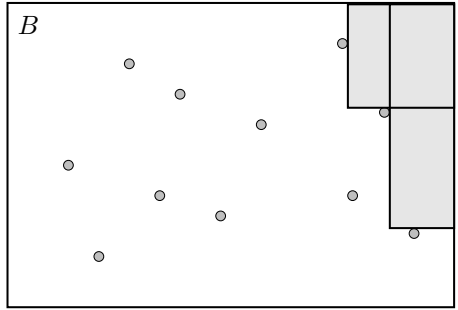
\begin{figure}[htb]
\begin{center}
\input{mx2a.pstex_t}
\caption{Maximal $P$-empty rectangles with two adjacent edges on $\bd B$.}
\label{mx2a}
\end{center}
\end{figure}

\paragraph{(iii) Two opposite edges of $R$ lie on two opposite edges of $B$.}
Suppose, without loss of generality, that the left and right edges of
$R$ lie on $e_\ell$ and $e_r$, respectively. In this case the top and
bottom edges of $R$ must be supported by two points of $P$,
consecutive in their $y$-order. Clearly, there are $O(n)$ such pairs,
and thus also $O(n)$ rectangles of this kind. Again, handling the top
and bottom edges of $B$ is done in a fully symmetric manner.
See Figure~\ref{mx2b}.

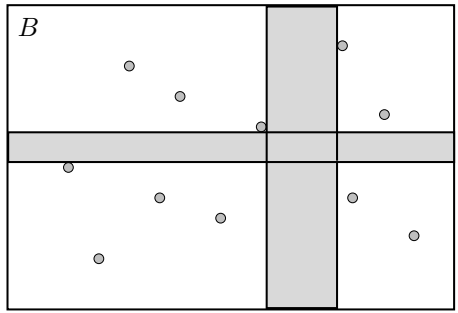
\begin{figure}[htb]
\begin{center}
\input{mx2b.pstex_t}
\caption{Maximal $P$-empty rectangles with two opposite edges on $\bd B$.}
\label{mx2b}
\end{center}
\end{figure}

\paragraph{(iv) One edge of $R$ lies on $\bd B$.}
Suppose, without loss of generality, that the right edge of $R$ lies
on $e_r$. Then the three other sides of $R$ must be supported by
points of $P$. For each point $p\in P$ there is a unique maximal
$P$-empty rectangle whose right edge lies on $e_r$ and whose left edge
passes through $p$. This rectangle is obtained by connecting $p$ to
$e_r$ by a horizontal segment $h$ and then by translating $h$ upwards
and downwards until it first hits two respective points of $P$, or
reaches $\bd B$. (In the latter situations we obtain rectangles of the
preceding types.) Hence there are $O(n)$ rectangles of this kind too.

\begin{figure}[htb]
\begin{center}
\input{mx1.pstex_t}
\caption{A maximal $P$-empty rectangle with one edge on $\bd B$.}
\label{mx1}
\end{center}
\end{figure}
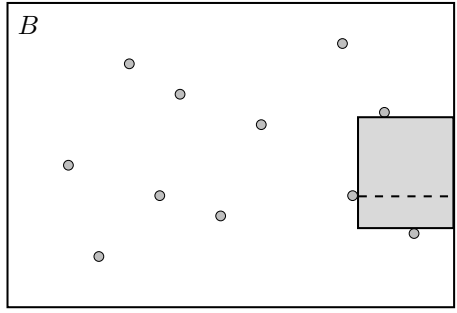

It is easy to compute all the maximal $P$-empty anchored rectangles
of the above four classes, in overall $O(n\log n)$ time: Computing
rectangles of type (i) and (iii) only requires sorting the points by
their $x$ or $y$ coordinates. Computing rectangles of type (ii) (of
the specific kind depicted in Figure \ref{mx2a}) requires computing
the list of maximal points. This can be done by scanning the points
from right to left maintaining the highest point seen so far. A
point $p$ is maximal if and only if it is higher than the previous
highest point. We can compute the rectangles of type (iv) (of the
specific kind depicted in Figure \ref{mx1}) also by traversing the
points from right to left while maintaining the points already
traversed, sorted by their $y$-coordinates, in a balanced search
tree. When we process a point $p$ then its successor and predecessor
(if they exist) in the tree define the top and bottom edges of the
rectangle of type (iv) whose left edge passes through $p$.

 We collect these
rectangles and store them in our two-dimensional segment tree $S$.
Given a query point $q$, we can find the rectangle of largest area
containing $q$ among these rectangles by searching in $S$, as
explained above, in $O(\log^2n)$ time.

\subsection{Maximal empty rectangles supported by four points of $P$}

In the remainder of the paper we are concerned only with maximal
$P$-empty rectangles supported by four points of $P$, one on each
side of the rectangle. We refer to such rectangles as \emph{bounded}
$P$-empty rectangles. We note that the number of such rectangles can
be $\Theta(n^2)$ in the worst case; see Figure~\ref{stair} for an
illustration of the lower bound.  The upper bound follows by
observing that there is at most one maximal $P$-empty rectangle
whose top and bottom edges pass through two respective specific
points of $P$. (To see this, take the rectangle having these points
as a pair of opposite vertices and, assuming it to be $P$-empty,
expand it to the left and to the right until its left and right
edges hit two additional respective points.) Handling these
(potentially quadratically many) rectangles has to be done
implicitly, in a manner that we now proceed to describe.

\begin{figure}[htb]
\begin{center}
\input{stair.pstex_t}
\caption{A set $P$ of $n$ points with $\Theta(n^2)$ maximal
$P$-empty rectangles.}
\label{stair}
\end{center}
\end{figure}
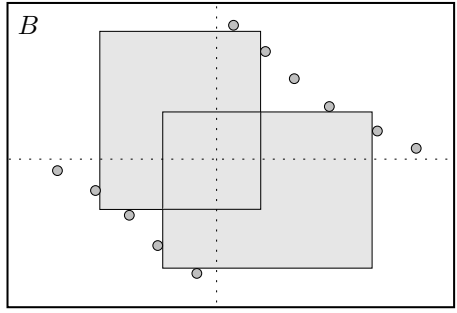

We store the points of $P$ in a two-dimensional range tree
(see, e.g., \cite{Dutchbook}).
The points are stored at the leaves of the primary tree $T$ in their
left-to-right order.  For a node $u$ of $T$, we denote by $P_u$ the
subset of the points stored at the leaves of the subtree rooted at
$u$. We associate with each internal node $u$ of $T$ a
{\em vertical splitter} $\ell_u$, which is a vertical line
separating the points stored at
the left subtree of $u$ from those stored at the right subtree.
These splitters induce a hierarchical binary decomposition of the plane
into vertical strips.  The strip $\sigma_{\rm root}$ associated with
the root is the entire plane, and the strip $\sigma_u$ of a node $u$
is the portion of the strip of the parent $p(u)$ of $u$ which is
delimited by $\ell_{p(u)}$ and contains $P_u$.

With each  node $u$ in $T$ we associate a secondary tree $T_u$
containing the points of $P_u$ in a bottom-to-top order. For a node
$v$ of $T_u$, we denote by $P_v$ the points stored at the leaves of
the subtree rooted at $v$. We associate with each internal node $v$
of $T_u$ a {\em horizontal splitter} $\ell_v$, which is a horizontal
line separating the points stored at the left subtree of $v$ from
those stored at the right subtree. These splitters induce a
hierarchical binary decomposition of the strip $\sigma_u$ into
rectangles.  The rectangle associated with the root of $T_u$ is the
entire vertical strip $\sigma_u$, and the rectangle $B_v$ of a node
$v$ is the portion of the rectangle of the parent $p(v)$ of $v$
which is delimited by $\ell_{p(v)}$ and contains $P_v$. See
Figure~\ref{boxv}.

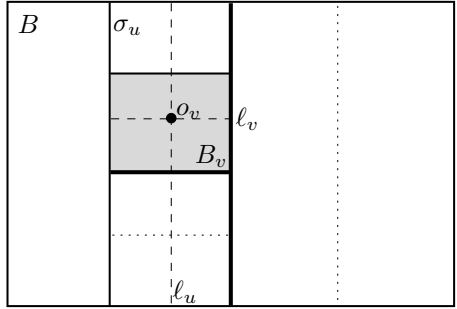
\begin{figure}[htb]
\begin{center}
\input{boxv.pstex_t}
\caption{The rectangle $B_v$ associated with a secondary node $v$,
with its splitters and origin.}
\label{boxv}
\end{center}
\end{figure}

In this way, the range tree defines a hierarchical subdivision of
the plane, so that each secondary node $v$ is associated with a
rectangular region $B_v$ of the subdivision. If $v$ is not a leaf
then it is associated with a horizontal splitter $\ell_v$. If the
primary node $u$ associated with the secondary tree of $v$ is also
not a leaf then $v$ is also associated with a vertical splitter
$\ell_u$. The vertical segment $\ell_u$ and the horizontal segment
$\ell_v$ meet at a point $o_v$ inside $B_v$, which we refer to as
the \emph{origin of $v$}.

A query point $q$ defines a search path $\pi_q$ in $T$ and a search
path in each secondary tree $T_u$ of a primary node $u$ on $\pi_q$.
We refer to the nodes on these $O(\log n)$ paths as constituting the
\emph{search set} of $q$, which therefore consists of $O(\log^2 n)$
secondary nodes.

Let $R$ be a bounded maximal $P$-empty rectangle containing $q$
supported by four points  $p_t$, $p_b$, $p_\ell$, and $p_r$ of $P$,
lying respectively on the top, bottom, left, and right edges of $R$.
Let $u$ be the lowest common ancestor of $p_\ell$ and $p_r$ in the
primary tree, and let $v$ be the lowest common ancestor of $p_t$ and
$p_b$ in $T_u$ (clearly, both $p_t$ and $p_b$ belong to $T_u$). By
construction, $R$ is contained in $B_v$ and contains both $q$ and
$o_v$. See Figure \ref{bddrect}. Note that both $v$ and $u$ are
internal nodes (each being an lowest common ancestor of two leaves)
so $o_v$ is indeed defined. Furthermore, one can easily verify that
$v$ is in the search set of $q$.

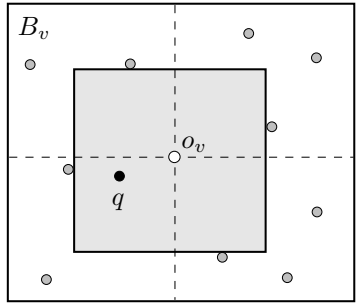
\begin{figure}[htb]
\begin{center}
\input{bddrect.pstex_t}
\caption{A bounded maximal $P$-empty rectangle of the subproblem at
$v$.} \label{bddrect}
\end{center}
\end{figure}

In the following we consider only secondary nodes $v$ which are not
leaves, and are associated with primary nodes $u$ which are not
leaves.

We define the \emph{subproblem at a secondary node $v$} (of the
above kind) as the problem of finding the largest-area bounded
maximal $P$-empty rectangle containing $q$ and $o_v$ which lies in
the interior of $B_v$. It follows that if we solve each subproblem
at each secondary node $v$ in the search set of $q$, and take the
rectangle of largest area among those subproblem outputs, we get the
largest-area bounded maximal $P$-empty rectangle containing $q$.

In the remainder of this section we focus on the solution of a
single subproblem at a node $v$ of a secondary tree $T_u$ in the
search set of $q$. We focus only on the points in $P_v$ and for
convenience we extend $B_v$ to the entire plane and we move $o_v$ to
the origin. The line $\ell_u$ becomes the $y$-axis, and the line
$\ell_v$ becomes the $x$-axis. Put $n_v = |P_v|$. We classify the
bounded maximal $P$-empty rectangles contained in $B_v$ and
containing the origin according to the quadrants containing the four
points associated with them, namely, those lying on their boundary,
and find the largest-area rectangle containing $q$ in each class
separately.

\paragraph{(i) Three defining points in a halfplane.}
The easy cases are when one of the four halfplanes defined by the
$x$-axis or the $y$-axis (originally $\ell_u$ and $\ell_v$) contains
three of the defining points. Suppose for specificity that this is
the halfplane to the left of the $y$-axis; the other four cases are
treated in a fully symmetric manner. See Figure~\ref{half3}.
Consider the subset $P_\ell$ of points to the left of the $y$-axis.
For each point $p$ of $P_\ell$ there is (at most) a single rectangle
in this family such that $p$ is its left defining point. Similarly
to the analysis in case (iv) of Section \ref{sec:boundary}, we
obtain this rectangle by connecting  $p$ to the $y$-axis by a
horizontal segment, and shifting this segment up and down until it
hits two other respective points of $P_\ell$. Now we have a
rectangle bounded by three points of $P_\ell$ whose right edge is
anchored to the $y$-axis. Extend this rectangle to the right until
its right edge hits a point to the right of the $y$-axis, to obtain
the unique rectangle of this type with $p$ on its left edge. (Here,
and in the other cases discussed below, we assume that all the
relevant points of $P$ do exist; otherwise, the rectangle that we
construct is not fully contained in the interior of $B_v$. This
would be the case, for example, if the shift of the above rectangle
to the right of the $y$-axis does not encounter any point of $P_v$.)
Clearly, there are $O(n_v)$ bounded maximal empty rectangles of this
type.

\begin{figure}[htb]
\begin{center}
\input{half3.pstex_t}
\caption{A bounded maximal $P$-empty rectangle with three defining
points to the left of the $y$-axis.}
\label{half3}
\end{center}
\end{figure}
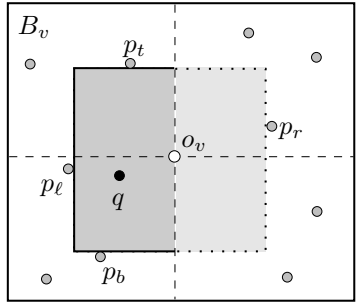

We can find the part of each such rectangle which is to the left of
the $y$-axis, by sweeping $B_v$ with a vertical line from the
$y$-axis to the left maintaining the points already seen in a
balanced search tree, exactly as we computed the empty rectangles of
type (iv) in Section \ref{sec:boundary}. To find the part of each
rectangle which is to the right of the $y$-axis we store the points
to the right of the $y$-axis, sorted by their $y$-coordinates, in a
 search tree $\Sigma_r$. With each node of $\Sigma_r$ we store
the leftmost point stored in its subtree. We can identify the right
edge of each rectangle $R$ by  using $\Sigma_r$ to find, in
logarithmic time, the leftmost point to the right of the $y$-axis
between the top and the bottom edges of $R$.

Overall we can find all rectangles of this type associated with $v$
in $O(n_v\log n_v)$ time.  Summing this cost over all secondary
nodes $v$, we obtain a total of $O(n\log^2 n)$ such rectangles,
which can be constructed in $O(n\log^3 n)$ overall time.

We add all these rectangles to the global segment tree $S$. The size
of the expanded tree $S$ remains $O(n\log n)$ since it still
suffices to store only the largest-area rectangle among all
rectangles associated with each secondary node. The preprocessing
time increases to $O(n\log^4 n)$ since each of the $O(n\log^2 n)$
rectangles is mapped to $O(\log^2 n)$ secondary nodes of $S$, and
for each rectangle $R$ and a node $u$ to which $R$ is mapped,  we
need to check whether $R$ is the largest rectangle mapped to $u$. A
query in $S$ still takes $O(\log^2 n)$ time.

\bigskip

The remaining cases involve bounded maximal $P$-empty rectangles $R$
such that each of the four halfplanes defined by the $y$-axis or the
$x$-axis contains exactly two defining points of $R$. This can
happen in two situations: either there exist two opposite quadrants,
each containing two defining points of $R$, or each quadrant
contains exactly one defining point of $R$.

\paragraph{(ii) One defining point in each quadrant.}
The situation in which each quadrant contains exactly one defining
point of $R$ is also easy to handle, because again there are only
$O(n_v)$  bounded maximal $P$-empty rectangles of this type in
$B_v$. To see this, consider, without loss of generality, the case
where the first quadrant contains the right defining point, $p_r$,
the second quadrant contains the top defining point, $p_t$, the
third quadrant contains the left defining point, $p_\ell$, and the
fourth quadrant contains the bottom defining point, $p_b$. See
Figure~\ref{quad1111}. (There is one other situation, in which the
top defining point lies in the first quadrant, the right point in
the fourth quadrant, the bottom point in the third, and the left
point in the second; this case is handled in a fully symmetric
manner.)

\begin{figure}[htb]
\begin{center}
\input{quad1111.pstex_t}
\caption{A bounded maximal $P$-empty rectangle with one defining
point in each quadrant.}
\label{quad1111}
\end{center}
\end{figure}
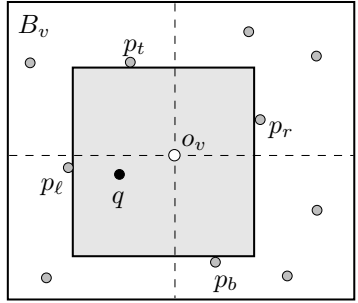

We claim that $p_r$ can be the right defining point of at most one
such rectangle. Indeed, if $p_r$ is the right defining point of such
a rectangle $R$ then $p_b$ is the first point we hit when we sweep
downwards a horizontal line segment connecting $p_r$ to the $y$-axis
(assuming the sweep reaches below the $x$-axis; otherwise $p_r$
cannot be the right defining point of any rectangle of the current
type). Similarly, the point $p_\ell$ is the first point that we hit
when we sweep to the left a vertical line segment connecting $p_b$
to the $x$-axis, $p_t$ is the first point we hit when we sweep
upwards a horizontal line segment connecting $p_\ell$ to the
$y$-axis, and finally $p_r$ is the first point we hit when we sweep
a vertical line segment connecting $p_t$ to the $x$-axis. As noted,
if any of the points we hit during these sweeps is not in the
correct quadrant, or the last sweep does not hit $p_r$ (e.g.,
because the point $p_t$ is lower than $p_r$), or one of the sweeps
does not hit any point before hitting $\bd B_v$, then $p_r$ is not
the right defining point of any rectangle of this type.

We compute these $O(n_v)$ bounded maximal $P$-empty rectangles using
four balanced search trees. As for rectangles of type (i) we
maintain the points to the right of the $y$-axis in a balanced
search tree $\Sigma_r$, sorted by their $y$ coordinates, storing
with each node the leftmost point in its subtree. Similarly, we
maintain the points below the $x$-axis in a balanced search tree
$\Sigma_b$ sorted by their $x$ coordinates, storing with each node
the topmost point in its subtree. We maintain the points to the left
of the $y$-axis and the points above the $x$-axis in symmetric
search trees $\Sigma_\ell$ and $\Sigma_t$, respectively. We can find
each rectangle in this family by four queries, starting with each
point $p_r$ in the first quadrant, first in $\Sigma_b$ to identify
$p_b$, then in $\Sigma_\ell$ to find $p_\ell$, in $\Sigma_t$ to find
$p_t$, and finally in $\Sigma_r$ to ensure that we get back to
$p_r$.

Summing over all secondary nodes $v$, we have $O(n\log^2 n)$ such
rectangles, which we can construct in $O(n\log^3 n)$ overall time.
We add them too to the global segment tree $S$, without changing the
asymptotic bounds on its performance parameters, as discussed
earlier.

\subsection{Two defining points in the first and third quadrants}

The hardest case is where, say, each of the first and third
quadrants contains two defining points of $R$.
(The case where each of the second and fourth quadrants contains
two defining points is handled symmetrically.) The defining points
in the first (resp., third) quadrant are consecutive minimal
(resp., maximal) points of the subset of $P_v$ in that quadrant.
There are $O(n_v)$ such pairs.
Denote the sequence of maximal points of the third quadrant by $E$,
and the sequence of minimal points of the first quadrant by $F$,
both sorted from left to right (or, equivalently, from top to bottom).

Consider a consecutive pair $(a,b)$ in $E$ (with $a$ to the left and
above $b$).  Let $M_1$ be the unique maximal $P$-empty rectangle whose
right edge is anchored at the $y$-axis, its left edge passes through
$a$, its bottom edge passes through $b$, and its top edge passes
through some point $c$ (in the second quadrant); it is possible that
$c$ does not exist, in which case some minor modifications (actually,
simplifications) need to be applied to the forthcoming analysis, which
we do not spell out.

Let $M_2$ be the unique maximal empty rectangle whose top edge is
anchored at the $x$-axis, its left edge passes through $a$, its
bottom edge passes through $b$, and its right edge passes through
some point $d$ (in the fourth quadrant; again, we ignore the case
where $d$ does not exist).  See Figure \ref{rectbd}.
Our maximal empty rectangle cannot
extend higher that $c$, nor can it extend to the right of $d$.
Hence its two other defining points must be a pair $(w,z)$ of
consecutive elements of $F$, both lying to the
left of $d$ and below $c$. The minimal points which satisfy these
constraints form a contiguous subsequence of $F$.

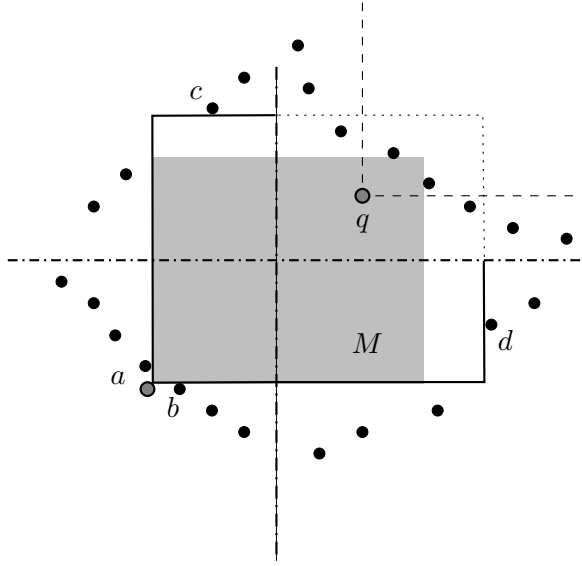
\begin{figure}[htb]
\begin{center}
\input{rectbd.pstex_t}
\caption{The structure of maximal empty rectangles with two defining
points in each of the first and third quadrants.}
\label{rectbd}
\end{center}
\end{figure}

That is, for each consecutive pair $\rho=(a,b)$ of points of $E$
we have a contiguous ``interval'' $I_\rho\subseteq F$,
so that any consecutive pair $\pi=(w,z)$ of points in
$I_\rho$ defines with $\rho$ a maximal empty rectangle which
contains the origin, and these are the only pairs which can define
with $\rho$ such a rectangle. (Note that we can ignore the ``extreme''
rectangles  defined by $a,b,c$, and the highest point of $I_\rho$, or by
$a,b,d$, and the lowest point of $I_\rho$, since these rectangles have
three of their defining points in a common halfplane defined by the
$x$-axis or by the $y$-axis, and have therefore already been treated.)

To answer queries with respect to these rectangles, we process the
data as follows. We compute the chain $E$ of maximal points in the
third quadrant and the chain $F$ of minimal points in the first
quadrant, ordered as above. This is done in $O(n_v \log n_v)$ time
in the same way as we computed the chain of maximal points of $P$ in
Section \ref{sec:boundary}.
 For each pair
$\rho=(a,b)$ of consecutive points in $E$ we compute the
corresponding delimiting points $c$ (in the second quadrant) and $d$
(in the fourth quadrant). Formally, $c$ is the lowest point in the
second quadrant which lies to the right of $a$, and $d$ is the
leftmost point in the fourth quadrant which lies above $b$. We then
use $c$ and $d$ to ``carve out'' the interval $I_\rho$ of $F$,
consisting of those points that lie below $c$ and to the left of
$d$. We can find $c$ by a binary search in the chain of $y$-minimal
and $x$-maximal points in the second quadrant, and find $d$ by a
binary search in the chain of $x$-minimal and $y$-maximal points in
the fourth quadrant. These chains can be computed in the same way as
in the construction of $E$ and $F$. Once we have the chains we can
find, for each consecutive pair $\rho=(a,b)$ in $E$, the
corresponding entities $c$, $d$, and $I_\rho$, in $O(\log n_v)$
time.

We next define a matrix $A$ as follows. Each row of $A$ corresponds
to a pair $\rho$ of consecutive points in $E$ and each column of $A$
corresponds to a pair $\pi$ of consecutive points in $F$. If at
least one point of $\pi$ is not in $I_\rho$ then the value of
$A_{\rho\pi}$ is undefined. Otherwise, it is equal to the area of
the (maximal empty) rectangle defined by $\rho$ and $\pi$. By the
preceding analysis, the defined entries in each row form a
contiguous subsequence of columns. It is easy to verify that if
$\rho_2$ follows (i.e., lies more to the right and below) $\rho_1$
on $E$ then the left (resp., right) endpoint of $I_{\rho_2}$ cannot
be to the right of the left (resp., right) endpoint of $I_{\rho_1}$;
See Figure~\ref{overlap}. It follows that in each column of $A$ the
defined entries also form a contiguous subsequence of rows. We refer
to such partially defined matrix $A$ in which the defined part of
each row and of each column is consecutive as a {\em double
staircase} matrix.

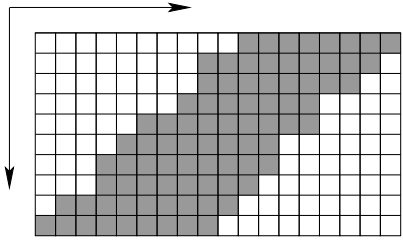
\begin{figure}[htb]
\begin{center}
\input{overlap.pstex_t}
\caption{The structure of the defined portion of $A$.}
\label{overlap}
\end{center}
\end{figure}

The following simple lemma plays a crucial role in our analysis.

\begin{lemma} \label{monge}
Let $x_1$, $x_2$, $y_1$, $y_2$ be four points in the plane, so that
$x_1$ and $x_2$ lie in the first quadrant, $y_1$ and $y_2$ lie in
the third quadrant, $x_1$ lies northwest to $x_2$, and $y_1$ lies
northwest to $y_2$. For any point $w$ in the third quadrant and any
point $z$ in the first quadrant, let $R(w,z)$ denote the rectangle
having $w$ and $z$ as opposite corners, and let $A(w,z)$ denote
the area of $R(w,z)$. Then we have
\begin{equation} \label{eq:monge}
A(y_1,x_1) + A(y_2,x_2) > A(y_1,x_2) + A(y_2,x_1) .
\end{equation}
\end{lemma}
\begin{proof}
The situation is depicted in Figure~\ref{shmonge}. In the notation
of the figure we have
$$
A(y_1,x_1) + A(y_2,x_2) = A(y_1,x_2) + A(y_2,x_1) + A_1 + A_2 ,
$$
where $A_1$ and $A_2$ are the areas of the two shaded rectangles.
\end{proof}

\begin{figure}[htb]
\begin{center}
\input{shmonge.pstex_t}
\caption{The inverse Monge property of maximal rectangles.}
\label{shmonge}
\end{center}
\end{figure}
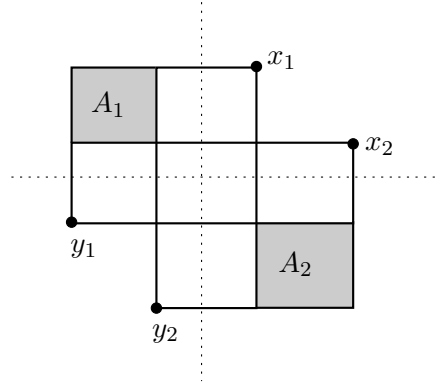

Lemma~\ref{monge} asserts that if $A_{\rho_1 \pi_1}$,
$A_{\rho_2 \pi_2}$, $A_{\rho_1 \pi_2}$, and $A_{\rho_2 \pi_2}$,
for $\rho_1 < \rho_2$ and $\pi_1 < \pi_2$, are all defined then
$$
A_{\rho_1 \pi_1} + A_{\rho_2 \pi_2} >
A_{\rho_1 \pi_2} + A_{\rho_2 \pi_1} ,
$$
or, equivalently,
\begin{equation} \label{eq:inv-monge}
A_{\rho_1 \pi_1} - A_{\rho_2 \pi_1} >
A_{\rho_1 \pi_2} - A_{\rho_2 \pi_2} .
\end{equation}
Hence $A$ satisfies the  inverse Monge property, with respect to its
defined entries, so it is a partial inverse Monge (and thus also
totally monotone)  matrix.

\subsubsection{Answering a query in the first (or third) quadrant}
\label{s13}

The next step is to compute the column maxima in $A$. That is, for
each pair $\pi$ of consecutive points in $F$, we compute the value
$A_\mx(\pi) = \max_\rho A_{\rho\pi}$, where $\rho$ ranges over all
consecutive pairs in $E$ for which $\pi$ is contained in $I_\rho$.
The fact that $A$ is only partially defined makes this task slightly
more involved than the similar task for totally defined inverse
Monge matrices. Intuitively, this computation is similar to the
construction of an \emph{upper envelope} of pseudo-segments in the
plane. Indeed, we can think of the entries of a particular row
$\rho$ as forming the graph of a (discrete) partially defined
function $\hat{A}_\rho(\cdot)$, mapping indices $\pi$ of columns to
the areas $A_{\rho\pi}$ of the corresponding rectangles. Equation
(\ref{eq:inv-monge}) implies that these functions behave as
pseudo-segments. Specifically,  we extend the  domain of definition
of each function $\hat{A}_\rho$ to a segment, delimited by the first
and last pairs $\pi$ at which $A_{\rho\pi}$ is defined, by linearly
interpolating between each pair of $\pi$-consecutive points on its
graph. Then (\ref{eq:inv-monge}) implies that each pair of the
resulting connected polygonal curves intersect at most once.

The complexity of the upper envelope of $m$ pseudo-segments is
$O(m\alpha(m))$ (see \cite{SA}). More precisely, this expression
bounds the number of \emph{breakpoints} of the envelope (points
where two distinct graphs intersect on the envelope; for technical
reasons we also regard the leftmost and rightmost points of each
graph as breakpoints), and ignores the complexity of the individual
functions (the graph of each of our functions consists of many
segments, one fewer than the number of columns where $A$ is defined
at the corresponding row, and these individual complexities are
ignored in the bound above); this comment is crucial for the
complexity analysis of our procedure. Since we can find the
intersection of any pair of pseudo-segments $\hat{A}_{\rho_1}$ and
$\hat{A}_{\rho_2}$ in $O(\log n)$ time, by a binary search through
the relevant columns, we can compute this upper envelope in
$O(m\alpha(m)\log m \log n)$ time, by a simple divide-and-conquer
algorithm, or in $O(m\log m\log n)$ time, using the more elaborate
algorithm of Hershberger~\cite{Hershberger89}.

In accordance with the remark made in the preceding paragraph, we
note that in the algorithm just sketched we do not attempt to
compute and output the upper envelope explicitly---this will take
$\Omega(n)$ time at each recursive step, for filling in the value of
the envelope at every column. Instead, we only compute its
\emph{breakpoints}, which partition the columns into $O(m\alpha
(m))$ blocks, so that for all $\pi$ in the same block, $\max_\rho
A_{\rho\pi}$ is attained by the same row. This implicit
representation is significantly cheaper when $m \ll n$, and is
crucial to obtain the running time asserted above.

\medskip

\noindent {\bf Remark.} Rather than adapting the divide-and-conquer
algorithm for upper envelopes,  just mentioned above, to our
discrete settings, we can use an algorithm of Klawe \cite{Klawe92}
for computing row maxima in staircase (inverse) Monge matrices.
Recall that a \emph{staircase matrix} is a partially defined matrix
in which the defined portion of each row is contiguous starting at
the leftmost column, and the defined part of each row is not smaller
than the defined part of the preceding row. We shall also refer as a
staircase matrix to a matrix that can be made staircase by inverting
the order of the rows and the columns. (Note that the operation of
inverting the order of the rows and the columns preserves the
inverse Monge property.)
 Klawe shows how to find
row maxima in  a staircase totally monotone matrix in $O(n\alpha(m)+
m)$ time, where $m$ is the number of rows and $n$ is the number of
columns.

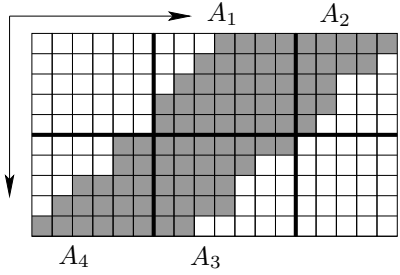
\begin{figure}[htb]
\begin{center}
\input{klawe.pstex_t}
\caption{The recursive construction using Klawe's algorithm.}
\label{klawe}
\end{center}
\end{figure}

We can use Klawe's algorithm as follows; refer to
Figure~\ref{klawe}. We split the matrix into four submatrices
$A_1,\ldots,A_4$ at the middle row, where $A_1$ and $A_2$ are formed
by the first half of the rows and $A_3$ and $A_4$ are formed by the
second half. We take the contiguous block of columns whose defined
portions intersect the middle row, split each of these columns at
the middle row, and form $A_1$ from the top parts of these columns
and $A_3$ from from the bottom parts. The submatrix $A_2$ (resp.,
$A_4$) is defined by the rows above (resp., below)  the middle row
and by the columns whose defined portions are fully contained in
this range of rows.
 Clearly, $A_1$ and $A_3$ are
two staircase submatrices (one straight and one inverted).

It follows that we can find column maxima in $A_1$ and $A_3$
 by two applications of Klawe's algorithm to $A_1^t$ and
$A_3^t$, and then by taking the maximum of the two relevant outputs
for each column. (Formally, Klawe's algorithm finds row maxima and
we need column maxima, but since the transpose operation preserves
the inverse Monge property, the application of Klawe's  algorithm to
the transposed matrix yields the desired column maxima.) We then
recursively apply the algorithm to the submatrices $A_2$, $A_4$.
Note that $A_2$ and $A_4$ are  disjoint submatrices of $A$, each
with half as many rows, and their column ranges are disjoint, and
also disjoint from the column ranges of $A_1$ and $A_3$. This is
easily seen to imply that the running time of this recursive
algorithm is $O(m  \alpha(n) \log m + n)$.

Although this algorithm is faster (when $n=O(m)$) than the one based
on computing the upper envelope of pseudo-segments, the latter will
be used again later, when we show how to handle query points in the
second or fourth quadrants.

\medskip

By construction, the upper envelope of the pseudo-segments
corresponding to the rows of $A$ records the column maxima of $A$.
Specifically, we scan the upper envelope from left to right, and the
maximum for each column occurs at the row that attains the upper
envelope at that column. (We can afford to perform this scan once,
upon termination of the whole procedure, but not at each of the
recursive steps of constructing sub-envelopes.)

After computing the column maxima in $A$ we build a range-maxima
data structure storing these column maxima, so that we can
efficiently retrieve the maximum in any query contiguous subsequence
of the columns. Such a structure can be constructed in time (and
storage) linear in the number of columns, and a query can be
answered in $O(1)$ time (see \cite{Bender:2000} and the reference
therein for the original results). For our purpose, though, since we
have to search $F$ to identify the interval of columns that the
query point $q$ ``controls'', we might as well use a standard binary
search tree over the columns, instead of the more sophisticated
structure of \cite{Bender:2000}.  We store in each subtree of the
tree the maximum of the column maxima, over all columns stored at
the subtree, which allows us to find the maximum in a query interval
of columns in $O(\log n)$ time.

The query point $q$ itself, if it lies in the first quadrant,
defines a contiguous subsequence $J_q$ of the sequence $F$ of
minimal points in the first quadrant, namely, those that lie above
$q$ and to its right. Only consecutive pairs within this subsequence
can form the top and right defining points of a maximal empty
rectangle containing $q$ of the type considered here. So we compute
$J_q$, in logarithmic time, and compute $\max_\pi A_\mx(\pi)$, over
all pairs $\pi$ contained in $J_q$, using the range-maxima data
structure just described, and output the corresponding rectangle.

As described so far we need two search trees, one is the
range-maxima data structure over the column maxima of $A$, and the
other is a search tree over $F$ which we use to identify the
subsequence $J_q$ carved out from $F$ by a query point $q$. Since a
column in $A$ corresponds to a pair of consecutive points of $F$ we
can in fact use only one search tree for both purposes. This search
tree is over the points in $F$ and it stores in each node $v$ the
largest of the  column maximum of the columns which are associated
with pairs of consecutive points in the subtree of $v$.


A query with a point in the third quadrant is handled in a fully
symmetric manner, using a symmetric data structure in which the roles
of $E$ and $F$ are interchanged. The cases where the query is in
the second or fourth quadrants will be considered next.

This  structure for queries in the first quadrant uses an additional
binary search tree
 for range maxima queries, for each secondary node in
$T$. The total size of these structures is $O(n\log^2 n)$, and they
can all be constructed in $O(n\log^4 n)$ overall time, using Hershberger's
algorithm \cite{Hershberger89} (In each secondary node $v$ we need $O(n_v\log^2 n_v)$ time
 to compute the upper
envelope of the functions defined by the rows of the matrix
associated with $v$.) A query takes $O(\log^3 n)$ time, because we
spend logarithmic time at each secondary node $v$ of $T$ on the
search paths, for which $q$ is in the first or third quadrant of
$B_v$.

\subsubsection{Answering a query in the second (or fourth) quadrant}
\label{s24}

Consider next the case where $q$ is in the second quadrant of $B_v$
(the case where $q$ is in the fourth quadrant is handled in a
symmetric manner). Consider the prefix $F_q$ of $F$ consisting of
points whose $y$-coordinate is larger than that of $q$, and the
prefix $E_q$ of $E$ consisting of points whose $x$-coordinate is
smaller than that of $q$. The rectangles defined by pairs of
consecutive points in $E$ and in $F$ which contain $q$ are exactly
those defined by pairs with at least one point in $E_q$ and one
point in $F_q$. See Figure~\ref{secquad}.

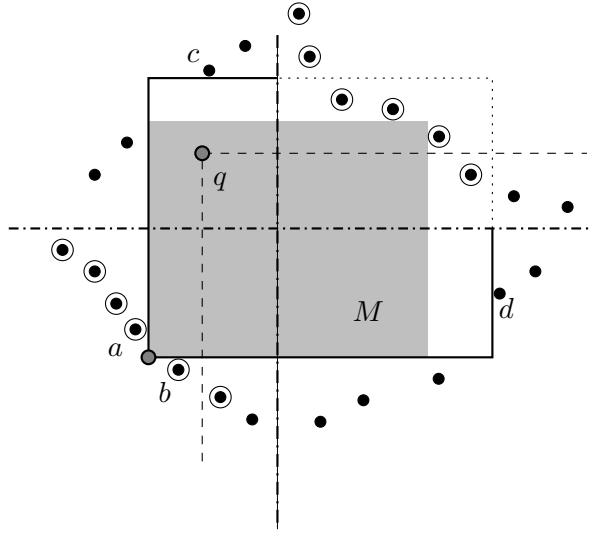
\begin{figure}[htb]
\begin{center}
\input{secquad.pstex_t}
\caption{Querying with a point in the second quadrant. The highlighted
points form the prefixes $E_q$ and $F_q$, plus one extra point in each
chain.}
\label{secquad}
\end{center}
\end{figure}

Here is an overview of our approach. We use the same matrix $A$
defined in the preceding subsection. We store the rows of $A$ in a
balanced binary search tree $T_h$. Each node $u$ of $T_h$ stores the
upper envelope $\E_u$ of the pseudo-segments corresponding to the
rows in the subtree of $u$. Given a query $q$ in the second
quadrant, we compute $E_q$, retrieve the pair $\rho_q$ formed by the
last point of $E_q$ and the next point of $E$ (if such a point
exists; otherwise we form the last pair in $E_q=E$), and represent
the first $\rho_q$ rows of $A$ as the disjoint union of $O(\log
n_v)$ canonical subsets of rows, corresponding to a collection $N_q$
of $O(\log n_v)$ respective nodes of $T_h$. See
Figure~\ref{canonenv} for a schematic depiction of this structure.

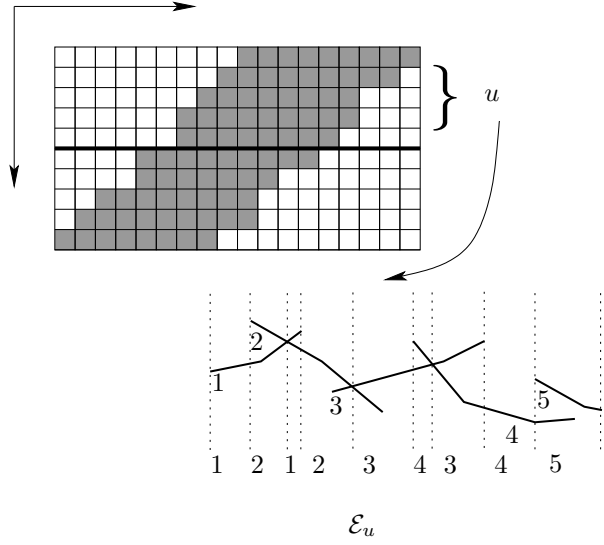
\begin{figure}[htb]
\begin{center}
\input{canonenv.pstex_t}
\caption{A single node $u$ of the tree $T_h$ and a schematic
representation of its upper envelope $\E_u$.}
\label{canonenv}
\end{center}
\end{figure}

We next compute $F_q$ and its ``last pair'' $\pi_q$, defined
analogously to $\rho_q$. What we need to do is to compute
$\max\{\E_u(\pi) \mid \pi\le\pi_q\}$, over all nodes $u\in N_q$, and
return the largest of these values (along with its corresponding
rectangle).

However, we cannot afford to enumerate the values of the  envelopes
$\E_u$ for all nodes $u\in T_h$ explicitly, because we may have
$\Theta(n_v)$ envelopes, each consisting of $\Theta(n_v)$ values, so
we may need quadratic storage for an explicit representation of the
envelopes. (This is the same problem that we faced in the preceding
subsection.) We therefore need an implicit representation that would
still allow us to compute the maximum of an envelope within a query
prefix range $\pi\le\pi_q$, in polylogarithmic time.

To do so, we use the compact representation of an envelope by its
breakpoints, as used in the preceding case. The divide-and-conquer
construction of the upper envelope of the entire range of rows of
$A$, described in the preceding subsection, yields as a by-product
all the upper envelopes $\E_u$, over all nodes $u$ of $T_h$. Each
envelope $\E_u$ is represented as a sequence of $O(m_u\alpha(m_u))$
intervals of columns, where $m_u$ is the number of rows stored at $u$
(the size of the subtree of $u$), so that, over each interval,
$\E_u$ is attained by some fixed row.

We thus face the following subproblem.  We are given an upper
envelope $\E_u$, defined at a node $u$ of $T_h$ which spans $m_u$
rows of $A$, as a sequence of $O(m_u\alpha(m_u))$ intervals of
columns delimited by breakpoints of $\E_u$ (as above some of these
breakpoints may be endpoints of the domains of definition of some
rows). Our goal is to preprocess $\E_u$ into a data structure so
that, given a prefix range of columns $\pi\le\pi_q$, we can compute
$\max\{\E_u(\pi) \mid \pi\le\pi_q\}$ efficiently. In the spirit of
the preceding discussion, the preprocessing has to take time that is
near-linear in $m$, and cannot afford an explicit enumeration of
$\E_u$. Instead, we use the following approach.

The compact representation of $\E_u$ calls for the design of a
black-box routine that receives as input a row $\rho$ and an
interval $[\pi_1,\pi_2]$ of columns, and returns $\max\{A_{\rho \pi}
\mid \pi_1\le \pi\le\pi_2\}$. Having such a routine at hand, we
first compute, by repeated calls to the black-box routine, the
maximum value of $\E_u$ over each of its $O(m_u\alpha(m_u))$
intervals (recalling that within each of these intervals $\E_u$ is
attained by a single row), then compute the cumulative maxima, for
each prefix of this sequence of intervals, and store these prefix
maxima in an array. Then, given a query index $\pi_q$, we find the
largest prefix of intervals that fully precede $\pi_q$, retrieve the
cumulative maximum of this prefix, and make one more call to the
black-box routine to retrieve $\max\{A_{\rho \pi} \mid \pi_\mx+1 \le
\pi\le\pi_q\}$, where $\pi_\mx$ is the index of the last column of
the last complete interval in the prefix, and where $\rho$ is the
row attaining the envelope $\E_u$ at the next interval. We return
the maximum of the output of this call and the retrieved prefix
maximum.

To implement this black-box routine, we apply a simpler variant of
the construction described so far, to the transposed matrix $A^t$.
That is, we store the rows of $A^t$ (originally, columns of $A$)
in a balanced binary tree $T_h^t$, and apply a divide-and-conquer
procedure for computing, for each node $w$ of $T_h^t$, the upper
envelopes $\E_w^t$ of the pseudo-segments corresponding to the
rows of $A^t$ stored at $w$, again, representing each envelope
$\E_w^t$ as a sequence of $O(m_w\alpha(m_w))$ intervals, where
$m_w$ is the size of the subtree of $w$. Now, given a query
$(\rho,\,[\pi_1,\pi_2])$, we search in $T_h^t$ and obtain a
representation of the interval $[\pi_1,\pi_2]$ as the disjoint
union of $O(\log n_v)$ canonical intervals of rows of $A^t$
(columns of $A$), corresponding to $O(\log n_v)$ nodes of $T_h^t$.
For each such node $w$, we retrieve $\E_w^t(\rho)$, in $O(\log m_w)$
time, by searching with $\rho$ through the sequence of intervals
representing $\E_w^t$. The maximum of these $O(\log n_v)$ values
$\E_w^t(\rho)$ is the desired maximum that we seek.
The overall cost of this computation is $O(\log^2 n_v)$.

To complete the analysis, we next bound the storage, preprocessing
cost, and query time for the entire structure.

Both trees $T_h$ and $T_h^t$ associated with a secondary node $v$ of
the range tree $T$ are of size $O(n_v \alpha(n_v)\log n_v)$,
including all secondary search trees over the upper envelopes, and
it takes a total of $O(n_v \alpha(n_v) \log^2 n_v )$ time to
construct them. (The divide-and-conquer algorithm does in fact
produce all the sub-envelopes at all the nodes of $T_h$ or of
$T_h^t$, at a particular secondary node $v$, within the above time
bound.)
 For each node $u\in T_h$ we
compute the maximum in each interval of $\E_u$ using $T_h^t$. This
takes $O(n_v \alpha(n_v)\log^3 n_v)$ time.

Summing over all secondary  nodes $v$ in $T$, we obtain that the
size of the entire range tree (including the respective trees $T_h$
and $T_h^t$ and a search tree over each envelope $\E_u$ (resp.,
$\E^t_u$) for each node $u$ in $T_h$ (resp., $T_h^t$)) is $O(n
\alpha(n)\log^3 n)$. The total preprocessing time is  $O(n
\alpha(n)\log^5 n)$ time, because the cost of processing each block
of each envelope $\E_u$, for nodes $u$ of $T_h$, is $O(\log^2 n_v)$,
using the black-box routine.

A query in each secondary node $v$ in which $q$ falls in the second
or fourth quadrant of $B_v$ takes $O(\log^3 n_v)$ time. This running
time follows since we need to perform one binary search in $T_h$ to
locate the $O(\log n_v)$ nodes $N_q$ representing the prefix $E_q$.
In each such node $u\in N_q$ we perform another binary search in
$\E_u$ to find the longest prefix of intervals that fully precede
$\pi_q$ and retrieve the cumulative maximum of this prefix. Finally
in each such node $u$ we make one more query to the data structure
representing $T_h^t$ to retrieve $\max\{A_{\rho \pi} \mid \pi_\mx+1
\le \pi\le\pi_q\}$, where $\pi_\mx$ is the index of the last column
of the last complete interval in the prefix of intervals fully
preceding $\pi_q$, and where $\rho$ is the row attaining the
envelope $\E_u$ at the next interval, which contains $\pi_q$. This
last query takes $O(\log^2 n_v)$ time, from which the overall cost
of $O(\log^3 n_v)$ at $v$ follows. This dominates the logarithmic
query cost at nodes $v$ where $q$ lies in the first or third
quadrant, and, summed over all secondary nodes $v$ of $T$, yields an
overall $O(\log^5 n)$ query time.

We recall that the entire presentation caters to maximal $P$-empty rectangles
having two defining points in the first quadrant of $B_v$ and two in the third
quadrant. To handle rectangles having two defining points in each of the
second and fourth quadrants, we prepare a second, symmetric version of
the structure in which the roles of quadrants are appropriately interchanged, and query both structures with $q$.

We can reduce the query time and the preprocessing time required at
a secondary node $v$ by a logarithmic factor using fractional
cascading \cite{CGfrac}. This technique allows us to insert bridges
between the envelopes corresponding to the nodes $w$ of  $T_h^t$ so
that once we locate the interval covering a particular column $\rho$
(of $A^t$) in $\E_w^t$ we could locate the interval containing
$\rho$ in the envelope $\E_{w'}^t$ of a node $w'$ adjacent to $w$ in
$O(1)$ time. This allows us to construct a data structure over
$T_h^t$ so the maximum in a particular row $\rho$ of $A$ and a range
of columns $[\pi_1,\pi_2]$ of $A$ can be found in $O(\log n_v)$ time
instead of $O(\log^2 n_v)$ time. This modification does not incur
any space overhead.

In summary, we obtain the following main result of the paper.

\begin{theorem} \label{main-bdd}
The data structure described above requires $O(n\alpha(n)\log^3 n)$
storage, and can be constructed in $O(n \alpha(n) \log^4 n)$ time.
Using the structure, one can find the largest-area  $P$-empty
rectangle contained in $B$ and containing a query point $q$ in
$O(\log^4 n)$ time.
\end{theorem}

\section{Submatrix maxima in totally monotone matrices}

Consider a partially defined totally monotone  $n \times n$
 matrix $A$ in
which the defined entries in each row are consecutive.

The range minima data structure that we associated with a secondary node
in Section \ref{s13} is in fact a general data structure for
preprocessing such a matrix $A$, in $O(n\log^2 n)$ time, so that we
can find the maximum of any row within an interval of columns in $O(1)$ time. (To get a constant query time
we need one of the more sophisticated range maxima data structures mentioned there.) The
size of this data structure is linear in $n$.

If  $A$ is a {\em double staircase matrix} (so the defined entries
of each column are also consecutive) then we showed in Section
\ref{s24} how to construct a data structure so that we can find the
maxima in any submatrix of $A$ defined by a prefix of the rows and a
prefix of the columns, in $O(\log^2 n)$ time. This data structure
takes $O(n\alpha(n) \log n)$ space and
 $O(n\alpha(n) \log^2 n)$ time to construct.

 The latter data structure can be easily extended so that it can
 find the maxima in any contiguous submatrix of $A$. The bounds remains the same.
 Since our application does not require  general submatrix queries, we
 leave out the details of this extension, which are straightforward.
Nevertheless, hoping that applications of this extended structure
will arise in the future, we state the result explicitly:
\begin{theorem}
Given a double-staircase totaly monotone $n \times n$ matrix $A$,
one can preprocess it, in $O(n\alpha(n) \log^2 n)$ time, into a data
structure of size $O(n\alpha(n)\log n)$, so that, given any
contiguous submatrix $B$ of $A$, the maximum entry of $B$ can be
computed in $O(\log^2 n)$ time.
\end{theorem}


\end{document}

%% file: mx0.pstex_t
\begin{picture}(0,0)%
\includegraphics{mx0.pstex}%
\end{picture}%
\setlength{\unitlength}{3355sp}%
\begingroup\makeatletter\ifx\SetFigFont\undefined%
\gdef\SetFigFont#1#2#3#4#5{%
  \reset@font\fontsize{#1}{#2pt}%
  \fontfamily{#3}\fontseries{#4}\fontshape{#5}%
  \selectfont}%
\fi\endgroup%
\begin{picture}(3344,2294)(579,-2183)
\put(676,-141){\makebox(0,0)[lb]{\smash{{\SetFigFont{10}{12.0}{\rmdefault}{\mddefault}{\updefault}{\color[rgb]{0,0,0}$B$}%
}}}}
\put(1371,-1401){\makebox(0,0)[lb]{\smash{{\SetFigFont{10}{12.0}{\rmdefault}{\mddefault}{\updefault}{\color[rgb]{0,0,0}$q$}%
}}}}
\end{picture}%

%% file: mx2a.pstex_t
\begin{picture}(0,0)%
\includegraphics{mx2a.pstex}%
\end{picture}%
\setlength{\unitlength}{3355sp}%
\begingroup\makeatletter\ifx\SetFigFont\undefined%
\gdef\SetFigFont#1#2#3#4#5{%
  \reset@font\fontsize{#1}{#2pt}%
  \fontfamily{#3}\fontseries{#4}\fontshape{#5}%
  \selectfont}%
\fi\endgroup%
\begin{picture}(3344,2294)(579,-2183)
\put(676,-141){\makebox(0,0)[lb]{\smash{{\SetFigFont{10}{12.0}{\rmdefault}{\mddefault}{\updefault}{\color[rgb]{0,0,0}$B$}%
}}}}
\end{picture}%

%% file: mx2b.pstex_t
\begin{picture}(0,0)%
\includegraphics{mx2b.pstex}%
\end{picture}%
\setlength{\unitlength}{3355sp}%
\begingroup\makeatletter\ifx\SetFigFont\undefined%
\gdef\SetFigFont#1#2#3#4#5{%
  \reset@font\fontsize{#1}{#2pt}%
  \fontfamily{#3}\fontseries{#4}\fontshape{#5}%
  \selectfont}%
\fi\endgroup%
\begin{picture}(3344,2294)(579,-2183)
\put(676,-141){\makebox(0,0)[lb]{\smash{{\SetFigFont{10}{12.0}{\rmdefault}{\mddefault}{\updefault}{\color[rgb]{0,0,0}$B$}%
}}}}
\end{picture}%

%% file: mx1.pstex_t
\begin{picture}(0,0)%
\includegraphics{mx1.pstex}%
\end{picture}%
\setlength{\unitlength}{3355sp}%
\begingroup\makeatletter\ifx\SetFigFont\undefined%
\gdef\SetFigFont#1#2#3#4#5{%
  \reset@font\fontsize{#1}{#2pt}%
  \fontfamily{#3}\fontseries{#4}\fontshape{#5}%
  \selectfont}%
\fi\endgroup%
\begin{picture}(3344,2294)(579,-2183)
\put(676,-141){\makebox(0,0)[lb]{\smash{{\SetFigFont{10}{12.0}{\rmdefault}{\mddefault}{\updefault}{\color[rgb]{0,0,0}$B$}%
}}}}
\end{picture}%

%% file: stair.pstex_t
\begin{picture}(0,0)%
\includegraphics{stair.pstex}%
\end{picture}%
\setlength{\unitlength}{3355sp}%
\begingroup\makeatletter\ifx\SetFigFont\undefined%
\gdef\SetFigFont#1#2#3#4#5{%
  \reset@font\fontsize{#1}{#2pt}%
  \fontfamily{#3}\fontseries{#4}\fontshape{#5}%
  \selectfont}%
\fi\endgroup%
\begin{picture}(3344,2294)(579,-2183)
\put(676,-141){\makebox(0,0)[lb]{\smash{{\SetFigFont{10}{12.0}{\rmdefault}{\mddefault}{\updefault}{\color[rgb]{0,0,0}$B$}%
}}}}
\end{picture}%

%% file: boxv.pstex_t
\begin{picture}(0,0)%
\includegraphics{boxv.pstex}%
\end{picture}%
\setlength{\unitlength}{3355sp}%
\begingroup\makeatletter\ifx\SetFigFont\undefined%
\gdef\SetFigFont#1#2#3#4#5{%
  \reset@font\fontsize{#1}{#2pt}%
  \fontfamily{#3}\fontseries{#4}\fontshape{#5}%
  \selectfont}%
\fi\endgroup%
\begin{picture}(3344,2301)(579,-2189)
\put(676,-141){\makebox(0,0)[lb]{\smash{{\SetFigFont{10}{12.0}{\rmdefault}{\mddefault}{\updefault}{\color[rgb]{0,0,0}$B$}%
}}}}
\put(1986,-1101){\makebox(0,0)[lb]{\smash{{\SetFigFont{10}{12.0}{\rmdefault}{\mddefault}{\updefault}{\color[rgb]{0,0,0}$B_v$}%
}}}}
\put(1846,-746){\makebox(0,0)[lb]{\smash{{\SetFigFont{10}{12.0}{\rmdefault}{\mddefault}{\updefault}{\color[rgb]{0,0,0}$o_v$}%
}}}}
\put(1821,-2106){\makebox(0,0)[lb]{\smash{{\SetFigFont{10}{12.0}{\rmdefault}{\mddefault}{\updefault}{\color[rgb]{0,0,0}$\ell_u$}%
}}}}
\put(1381,-126){\makebox(0,0)[lb]{\smash{{\SetFigFont{10}{12.0}{\rmdefault}{\mddefault}{\updefault}{\color[rgb]{0,0,0}$\sigma_u$}%
}}}}
\put(2286,-831){\makebox(0,0)[lb]{\smash{{\SetFigFont{10}{12.0}{\rmdefault}{\mddefault}{\updefault}{\color[rgb]{0,0,0}$\ell_v$}%
}}}}
\end{picture}%

%% file: bddrect.pstex_t
\begin{picture}(0,0)%
\includegraphics{bddrect.pstex}%
\end{picture}%
\setlength{\unitlength}{3355sp}%
\begingroup\makeatletter\ifx\SetFigFont\undefined%
\gdef\SetFigFont#1#2#3#4#5{%
  \reset@font\fontsize{#1}{#2pt}%
  \fontfamily{#3}\fontseries{#4}\fontshape{#5}%
  \selectfont}%
\fi\endgroup%
\begin{picture}(2604,2244)(579,-2133)
\put(676,-141){\makebox(0,0)[lb]{\smash{{\SetFigFont{10}{12.0}{\rmdefault}{\mddefault}{\updefault}{\color[rgb]{0,0,0}$B_v$}%
}}}}
\put(1371,-1401){\makebox(0,0)[lb]{\smash{{\SetFigFont{10}{12.0}{\rmdefault}{\mddefault}{\updefault}{\color[rgb]{0,0,0}$q$}%
}}}}
\put(1886,-976){\makebox(0,0)[lb]{\smash{{\SetFigFont{10}{12.0}{\rmdefault}{\mddefault}{\updefault}{\color[rgb]{0,0,0}$o_v$}%
}}}}
\end{picture}%

%% file: half3.pstex_t
\begin{picture}(0,0)%
\includegraphics{half3.pstex}%
\end{picture}%
\setlength{\unitlength}{3355sp}%
\begingroup\makeatletter\ifx\SetFigFont\undefined%
\gdef\SetFigFont#1#2#3#4#5{%
  \reset@font\fontsize{#1}{#2pt}%
  \fontfamily{#3}\fontseries{#4}\fontshape{#5}%
  \selectfont}%
\fi\endgroup%
\begin{picture}(2604,2244)(579,-2133)
\put(676,-141){\makebox(0,0)[lb]{\smash{{\SetFigFont{10}{12.0}{\rmdefault}{\mddefault}{\updefault}{\color[rgb]{0,0,0}$B_v$}%
}}}}
\put(1371,-1401){\makebox(0,0)[lb]{\smash{{\SetFigFont{10}{12.0}{\rmdefault}{\mddefault}{\updefault}{\color[rgb]{0,0,0}$q$}%
}}}}
\put(846,-1301){\makebox(0,0)[lb]{\smash{{\SetFigFont{10}{12.0}{\rmdefault}{\mddefault}{\updefault}{\color[rgb]{0,0,0}$p_\ell$}%
}}}}
\put(1886,-976){\makebox(0,0)[lb]{\smash{{\SetFigFont{10}{12.0}{\rmdefault}{\mddefault}{\updefault}{\color[rgb]{0,0,0}$o_v$}%
}}}}
\put(1456,-271){\makebox(0,0)[lb]{\smash{{\SetFigFont{10}{12.0}{\rmdefault}{\mddefault}{\updefault}{\color[rgb]{0,0,0}$p_t$}%
}}}}
\put(2606,-876){\makebox(0,0)[lb]{\smash{{\SetFigFont{10}{12.0}{\rmdefault}{\mddefault}{\updefault}{\color[rgb]{0,0,0}$p_r$}%
}}}}
\put(1301,-1951){\makebox(0,0)[lb]{\smash{{\SetFigFont{10}{12.0}{\rmdefault}{\mddefault}{\updefault}{\color[rgb]{0,0,0}$p_b$}%
}}}}
\end{picture}%

%% file: quad1111.pstex_t
\begin{picture}(0,0)%
\includegraphics{quad1111.pstex}%
\end{picture}%
\setlength{\unitlength}{3355sp}%
\begingroup\makeatletter\ifx\SetFigFont\undefined%
\gdef\SetFigFont#1#2#3#4#5{%
  \reset@font\fontsize{#1}{#2pt}%
  \fontfamily{#3}\fontseries{#4}\fontshape{#5}%
  \selectfont}%
\fi\endgroup%
\begin{picture}(2604,2244)(579,-2133)
\put(676,-141){\makebox(0,0)[lb]{\smash{{\SetFigFont{10}{12.0}{\rmdefault}{\mddefault}{\updefault}{\color[rgb]{0,0,0}$B_v$}%
}}}}
\put(1371,-1401){\makebox(0,0)[lb]{\smash{{\SetFigFont{10}{12.0}{\rmdefault}{\mddefault}{\updefault}{\color[rgb]{0,0,0}$q$}%
}}}}
\put(2531,-881){\makebox(0,0)[lb]{\smash{{\SetFigFont{10}{12.0}{\rmdefault}{\mddefault}{\updefault}{\color[rgb]{0,0,0}$p_r$}%
}}}}
\put(1886,-976){\makebox(0,0)[lb]{\smash{{\SetFigFont{10}{12.0}{\rmdefault}{\mddefault}{\updefault}{\color[rgb]{0,0,0}$o_v$}%
}}}}
\put(1456,-271){\makebox(0,0)[lb]{\smash{{\SetFigFont{10}{12.0}{\rmdefault}{\mddefault}{\updefault}{\color[rgb]{0,0,0}$p_t$}%
}}}}
\put(846,-1301){\makebox(0,0)[lb]{\smash{{\SetFigFont{10}{12.0}{\rmdefault}{\mddefault}{\updefault}{\color[rgb]{0,0,0}$p_\ell$}%
}}}}
\put(2126,-2016){\makebox(0,0)[lb]{\smash{{\SetFigFont{10}{12.0}{\rmdefault}{\mddefault}{\updefault}{\color[rgb]{0,0,0}$p_b$}%
}}}}
\end{picture}%

%% file: rectbd.pstex_t
\begin{picture}(0,0)%
\includegraphics{rectbd.pstex}%
\end{picture}%
\setlength{\unitlength}{3552sp}%
\begingroup\makeatletter\ifx\SetFigFont\undefined%
\gdef\SetFigFont#1#2#3#4#5{%
  \reset@font\fontsize{#1}{#2pt}%
  \fontfamily{#3}\fontseries{#4}\fontshape{#5}%
  \selectfont}%
\fi\endgroup%
\begin{picture}(4094,3934)(54,-3383)
\put(1344,-136){\makebox(0,0)[lb]{\smash{{\SetFigFont{11}{13.2}{\rmdefault}{\mddefault}{\updefault}{\color[rgb]{0,0,0}$c$}%
}}}}
\put(3496,-1884){\makebox(0,0)[lb]{\smash{{\SetFigFont{11}{13.2}{\rmdefault}{\mddefault}{\updefault}{\color[rgb]{0,0,0}$d$}%
}}}}
\put(2506,-1028){\makebox(0,0)[lb]{\smash{{\SetFigFont{11}{13.2}{\rmdefault}{\mddefault}{\updefault}{\color[rgb]{0,0,0}$q$}%
}}}}
\put(1186,-2355){\makebox(0,0)[lb]{\smash{{\SetFigFont{11}{13.2}{\rmdefault}{\mddefault}{\updefault}{\color[rgb]{0,0,0}$b$}%
}}}}
\put(2476,-1906){\makebox(0,0)[lb]{\smash{{\SetFigFont{11}{13.2}{\rmdefault}{\mddefault}{\updefault}{\color[rgb]{0,0,0}$M$}%
}}}}
\put(791,-2128){\makebox(0,0)[lb]{\smash{{\SetFigFont{11}{13.2}{\rmdefault}{\mddefault}{\updefault}{\color[rgb]{0,0,0}$a$}%
}}}}
\end{picture}%

%% file: overlap.pstex_t
\begin{picture}(0,0)%
\includegraphics{overlap.pstex}%
\end{picture}%
\setlength{\unitlength}{3355sp}%
\begingroup\makeatletter\ifx\SetFigFont\undefined%
\gdef\SetFigFont#1#2#3#4#5{%
  \reset@font\fontsize{#1}{#2pt}%
  \fontfamily{#3}\fontseries{#4}\fontshape{#5}%
  \selectfont}%
\fi\endgroup%
\begin{picture}(2945,1742)(435,-2608)
\end{picture}%

%% file: shmonge.pstex_t
\begin{picture}(0,0)%
\includegraphics{shmonge.pstex}%
\end{picture}%
\setlength{\unitlength}{3552sp}%
\begingroup\makeatletter\ifx\SetFigFont\undefined%
\gdef\SetFigFont#1#2#3#4#5{%
  \reset@font\fontsize{#1}{#2pt}%
  \fontfamily{#3}\fontseries{#4}\fontshape{#5}%
  \selectfont}%
\fi\endgroup%
\begin{picture}(3076,2686)(454,-3450)
\put(2266,-1225){\makebox(0,0)[lb]{\smash{{\SetFigFont{11}{13.2}{\rmdefault}{\mddefault}{\updefault}{\color[rgb]{0,0,0}$x_1$}%
}}}}
\put(2949,-1826){\makebox(0,0)[lb]{\smash{{\SetFigFont{11}{13.2}{\rmdefault}{\mddefault}{\updefault}{\color[rgb]{0,0,0}$x_2$}%
}}}}
\put(893,-2538){\makebox(0,0)[lb]{\smash{{\SetFigFont{11}{13.2}{\rmdefault}{\mddefault}{\updefault}{\color[rgb]{0,0,0}$y_1$}%
}}}}
\put(1463,-3131){\makebox(0,0)[lb]{\smash{{\SetFigFont{11}{13.2}{\rmdefault}{\mddefault}{\updefault}{\color[rgb]{0,0,0}$y_2$}%
}}}}
\put(1036,-1561){\makebox(0,0)[lb]{\smash{{\SetFigFont{11}{13.2}{\rmdefault}{\mddefault}{\updefault}{\color[rgb]{0,0,0}$A_1$}%
}}}}
\put(2346,-2676){\makebox(0,0)[lb]{\smash{{\SetFigFont{11}{13.2}{\rmdefault}{\mddefault}{\updefault}{\color[rgb]{0,0,0}$A_2$}%
}}}}
\end{picture}%

%% file: klawe.pstex_t
\begin{picture}(0,0)%
\includegraphics{klawe.pstex}%
\end{picture}%
\setlength{\unitlength}{3355sp}%
\begingroup\makeatletter\ifx\SetFigFont\undefined%
\gdef\SetFigFont#1#2#3#4#5{%
  \reset@font\fontsize{#1}{#2pt}%
  \fontfamily{#3}\fontseries{#4}\fontshape{#5}%
  \selectfont}%
\fi\endgroup%
\begin{picture}(2940,2044)(461,-2909)
\put(863,-2829){\makebox(0,0)[lb]{\smash{{\SetFigFont{10}{12.0}{\rmdefault}{\mddefault}{\updefault}{\color[rgb]{0,0,0}$A_4$}%
}}}}
\put(2791,-1044){\makebox(0,0)[lb]{\smash{{\SetFigFont{10}{12.0}{\rmdefault}{\mddefault}{\updefault}{\color[rgb]{0,0,0}$A_2$}%
}}}}
\put(1838,-2836){\makebox(0,0)[lb]{\smash{{\SetFigFont{10}{12.0}{\rmdefault}{\mddefault}{\updefault}{\color[rgb]{0,0,0}$A_3$}%
}}}}
\put(1959,-1036){\makebox(0,0)[lb]{\smash{{\SetFigFont{10}{12.0}{\rmdefault}{\mddefault}{\updefault}{\color[rgb]{0,0,0}$A_1$}%
}}}}
\end{picture}%

%% file: secquad.pstex_t
\begin{picture}(0,0)%
\includegraphics{secquad.pstex}%
\end{picture}%
\setlength{\unitlength}{3552sp}%
\begingroup\makeatletter\ifx\SetFigFont\undefined%
\gdef\SetFigFont#1#2#3#4#5{%
  \reset@font\fontsize{#1}{#2pt}%
  \fontfamily{#3}\fontseries{#4}\fontshape{#5}%
  \selectfont}%
\fi\endgroup%
\begin{picture}(4094,3705)(54,-3383)
\put(3496,-1884){\makebox(0,0)[lb]{\smash{{\SetFigFont{11}{13.2}{\rmdefault}{\mddefault}{\updefault}{\color[rgb]{0,0,0}$d$}%
}}}}
\put(2476,-1906){\makebox(0,0)[lb]{\smash{{\SetFigFont{11}{13.2}{\rmdefault}{\mddefault}{\updefault}{\color[rgb]{0,0,0}$M$}%
}}}}
\put(1501,-961){\makebox(0,0)[lb]{\smash{{\SetFigFont{11}{13.2}{\rmdefault}{\mddefault}{\updefault}{\color[rgb]{0,0,0}$q$}%
}}}}
\put(1318,-99){\makebox(0,0)[lb]{\smash{{\SetFigFont{11}{13.2}{\rmdefault}{\mddefault}{\updefault}{\color[rgb]{0,0,0}$c$}%
}}}}
\put(767,-2149){\makebox(0,0)[lb]{\smash{{\SetFigFont{11}{13.2}{\rmdefault}{\mddefault}{\updefault}{\color[rgb]{0,0,0}$a$}%
}}}}
\put(1122,-2469){\makebox(0,0)[lb]{\smash{{\SetFigFont{11}{13.2}{\rmdefault}{\mddefault}{\updefault}{\color[rgb]{0,0,0}$b$}%
}}}}
\end{picture}%

%% file: canonenv.pstex_t
\begin{picture}(0,0)%
\includegraphics{canonenv.pstex}%
\end{picture}%
\setlength{\unitlength}{3355sp}%
\begingroup\makeatletter\ifx\SetFigFont\undefined%
\gdef\SetFigFont#1#2#3#4#5{%
  \reset@font\fontsize{#1}{#2pt}%
  \fontfamily{#3}\fontseries{#4}\fontshape{#5}%
  \selectfont}%
\fi\endgroup%
\begin{picture}(4407,4005)(461,-5074)
\put(2831,-4121){\makebox(0,0)[lb]{\smash{{\SetFigFont{10}{12.0}{\rmdefault}{\mddefault}{\updefault}{\color[rgb]{0,0,0}$3$}%
}}}}
\put(4131,-4336){\makebox(0,0)[lb]{\smash{{\SetFigFont{10}{12.0}{\rmdefault}{\mddefault}{\updefault}{\color[rgb]{0,0,0}$4$}%
}}}}
\put(1961,-3956){\makebox(0,0)[lb]{\smash{{\SetFigFont{10}{12.0}{\rmdefault}{\mddefault}{\updefault}{\color[rgb]{0,0,0}$1$}%
}}}}
\put(2246,-3651){\makebox(0,0)[lb]{\smash{{\SetFigFont{10}{12.0}{\rmdefault}{\mddefault}{\updefault}{\color[rgb]{0,0,0}$2$}%
}}}}
\put(1951,-4561){\makebox(0,0)[lb]{\smash{{\SetFigFont{10}{12.0}{\rmdefault}{\mddefault}{\updefault}{\color[rgb]{0,0,0}$1$}%
}}}}
\put(2251,-4561){\makebox(0,0)[lb]{\smash{{\SetFigFont{10}{12.0}{\rmdefault}{\mddefault}{\updefault}{\color[rgb]{0,0,0}$2$}%
}}}}
\put(2701,-4561){\makebox(0,0)[lb]{\smash{{\SetFigFont{10}{12.0}{\rmdefault}{\mddefault}{\updefault}{\color[rgb]{0,0,0}$2$}%
}}}}
\put(3076,-4561){\makebox(0,0)[lb]{\smash{{\SetFigFont{10}{12.0}{\rmdefault}{\mddefault}{\updefault}{\color[rgb]{0,0,0}$3$}%
}}}}
\put(3451,-4561){\makebox(0,0)[lb]{\smash{{\SetFigFont{10}{12.0}{\rmdefault}{\mddefault}{\updefault}{\color[rgb]{0,0,0}$4$}%
}}}}
\put(3676,-4561){\makebox(0,0)[lb]{\smash{{\SetFigFont{10}{12.0}{\rmdefault}{\mddefault}{\updefault}{\color[rgb]{0,0,0}$3$}%
}}}}
\put(4051,-4561){\makebox(0,0)[lb]{\smash{{\SetFigFont{10}{12.0}{\rmdefault}{\mddefault}{\updefault}{\color[rgb]{0,0,0}$4$}%
}}}}
\put(2496,-4561){\makebox(0,0)[lb]{\smash{{\SetFigFont{10}{12.0}{\rmdefault}{\mddefault}{\updefault}{\color[rgb]{0,0,0}$1$}%
}}}}
\put(4366,-4081){\makebox(0,0)[lb]{\smash{{\SetFigFont{10}{12.0}{\rmdefault}{\mddefault}{\updefault}{\color[rgb]{0,0,0}$5$}%
}}}}
\put(4456,-4556){\makebox(0,0)[lb]{\smash{{\SetFigFont{10}{12.0}{\rmdefault}{\mddefault}{\updefault}{\color[rgb]{0,0,0}$5$}%
}}}}
\put(2971,-5001){\makebox(0,0)[lb]{\smash{{\SetFigFont{10}{12.0}{\rmdefault}{\mddefault}{\updefault}{\color[rgb]{0,0,0}${\cal E}_u$}%
}}}}
\put(3563,-1906){\makebox(0,0)[lb]{\smash{{\SetFigFont{34}{40.8}{\rmdefault}{\mddefault}{\updefault}{\color[rgb]{0,0,0}$\}$}%
}}}}
\put(3976,-1808){\makebox(0,0)[lb]{\smash{{\SetFigFont{10}{12.0}{\rmdefault}{\mddefault}{\updefault}{\color[rgb]{0,0,0}$u$}%
}}}}
\end{picture}%